\documentclass{amsart}
\usepackage{amsmath,amssymb}
\usepackage{graphicx,color}

\newtheorem{theorem}{Theorem}%[section]
\newtheorem{proposition}[theorem]{Proposition}

\newtheorem{lemma}[theorem]{Lemma}

\theoremstyle{remark}

\theoremstyle{definition}

\def\th{^{\text{th}}}
\def\tS{\widetilde{S}}
\DeclareMathOperator\tr{tr}

\def\bp{\bar p}
\def\bP{\bar P}

\newcommand{\al}{\alpha}
\newcommand{\be}{\beta}
\newcommand{\de}{\delta}
\newcommand{\ep}{\varepsilon}

\newcommand{\ga}{\gamma}

\def\RR{\mathbb{R}}

\newcommand\sub[1]{_{\mathrm{#1}}}

\renewcommand\leq\leqslant
\renewcommand\geq\geqslant
\newcommand\pd\partial

\begin{document}
\title[Uniqueness and characterization theorems]{Uniqueness and characterization theorems\\ for generalized entropies}
\date{Febrary 1, 2017}
\author{Alberto Enciso}
\address{Instituto de Ciencias Matem\'aticas, Consejo Superior de
  Investigaciones Cient\'\i ficas, 28049 Madrid, Spain}
\email{aenciso@icmat.es}

\author{Piergiulio Tempesta}
\address{Departamento de F\'\i sica Te\'orica II, Universidad
  Complutense de Madrid, 28040 Madrid, Spain,
and Instituto de Ciencias Matem\'aticas (CSIC-UAM-UCM-UC3M), 28049 Madrid, Spain}
\email{p.tempesta@fis.ucm.es}
\maketitle

\begin{abstract}
The requirement that an entropy function be composable is key: it means that the entropy of a compound system can be
calculated in terms of the entropy of its independent components.
We prove that, under mild regularity assumptions, the only composable
generalized entropy in trace form is the Tsallis one-parameter family
(which contains
Boltzmann--Gibbs as a particular case).

This result leads to the use of generalized entropies that are
not of trace form, such as R\'enyi's entropy, in the study of complex systems.
In this direction, we also present a characterization theorem for a
large class of composable non-trace-form entropy functions with
features akin to those of R\'enyi's entropy.

%This means that, in a way, the class of trace-form entropies is not suitable to construct new entropy functions if one requires that the entropy of a total system be determined by the those of its independent subsystems.
\end{abstract}

\tableofcontents

%\date

\section{Introduction}

The first example of a generalized entropy in information theory goes
back to the pioneering work of R\'enyi~\cite{Renyi1960} in the
early '60s, where he introduced a one-parameter family of entropies
that reduces to the classical Shannon entropy for a concrete value of
the parameter. R\'enyi was interested in the form of the most general
information measure satisfying certain natural requirements, in
particular additivity with respect to the composition of independent
statistical systems. The entropy introduced by Tsallis in 1988 in \cite{Tsallis} has been widely
investigated too: it is the \textit{non-additive} entropic form most
widely used in the natural and social sciences (see \cite{Tsbib} for an updated bibliography on the issue).

After these works, many different generalized entropic functions have
been constructed as non-additive information measures of a statistical
system (see, e.g., \cite{BR98}, \cite{Kaniadakis02}, \cite{Abe03},
\cite{Naudts08}, \cite{Hanel11}). The main motivation for these new
entropic forms lies, generally speaking, in the emergence of the
theory of complex systems, which often exhibits new phenomena that
require novel carefully designed information-theoretical tools for their interpretation.
For instance, several entropies other than Boltzmann--Gibbs have played a relevant role in the study of quantum entanglement \cite{CC1,CC2,RC1}, in the theory
of divergences generalizing the classical Kullback-Leibler one
\cite{KL,EH}, in geometric information theory \cite{AN}, and in theoretical linguistics and social sciences.

Several approaches have been proposed to classify the plethora of
entropy functions that have appeared in the literature over the last
decades. The standard axiomatic approach is based on the work by
Shannon \cite{Shannon} and Khinchin \cite{Khinchin}, who characterized
the Boltzmann entropy (within the class of trace-class entropies, which
we will define later) in terms of four requirements, now called the
Shannon--Khinchin (SK) axioms. Essentially, these axioms correspond to
the hypotheses that entropy, as a function defined on a certain space
of probability distributions, be continuous (SK1), expansible (i.e.,
adding an event of zero probability does not change the entropy)
(SK2), and that the uniform distribution maximizes the entropy
(SK3). The fourth axiom (SK4) is the additivity of the entropy, that
is, that the entropy of the composition of two subsystems
is the sum of their individual entropies.

It stands to reason that if one relaxes the additivity condition in axiom
(SK4), new possible functional forms of the entropy may arise. A
convenient way of doing this is to replace the
axiom~(SK4) by the weaker \textit{composability axiom}\/, introduced
in~\cite{PAOP16}. Roughly speaking, this axiom asserts that the entropy $S$ of a
compound system $A\cup B$ consisting of two \textit{independent}
systems~$A$ and~$B$ should be computable just in terms of the individual
entropies of~$A$ and~$B$. Using the notation $S(A)$ to represent the entropy of the
system~$A$, this means that there is a function of two variables, $\Phi(x,y)$, such that
\begin{equation}\label{composable}
S(A\cup B)=\Phi(S(A),S(B))
\end{equation}
for any independent systems  $A$ and $B$. Property~\eqref{composable}
is of fundamental importance; indeed, as in the case of the Boltzmann--Gibbs entropy,
it implies that an entropic function is properly defined on
macroscopic states of a given system, so that it can be computed
without knowing any information on the underlying microscopical
dynamics. This is the reason for which this property is key to ensure that the
entropy is physically meaningful. It should be mentioned that the
composability of an entropy also has tangible consequences from
an information-theoretical point of view \cite{ST16}.

Let us also recall the classical result by Lieb and Yngvason
\cite{LY99}, where the existence of an entropy function is derived
from monotonicity, additivity and extensivity requirements for all
allowed states. In the proof, the fundamental thermodynamic meaning of the
composability of classical entropy is laid bare.

A priori, the requirement that an entropy by composable is actually
even stronger than it looks. What we mean by this is that one must take into account that the combination of statistical systems (which is customarily
represented as~$A\cup B$ but is typically given by a tensor product of vector spaces) is \textit{associative} and \textit{commutative}\/. In addition, if we compound a system with another one in a zero-entropy configuration, the entropy of the compound system should be equal to the entropy of the first system. Therefore, it is crucial to demand that
\begin{eqnarray}\label{group-structure}
\nonumber \Phi(x,y)=\Phi(y,x) \qquad \Phi(x,0)=x \\
\Phi(x, \Phi(y,z))= \Phi(\Phi(x,y),z)
\end{eqnarray}
When these properties are satisfied, $S$ will be said to be a
(strictly) \textit{composable} generalized entropy \cite{PAOP16,
  PPROCA16}. If the composability axiom is satisfied only when the
subsystems are described by the uniform distribution, then we shall
say that the entropy $S$ is \textit{weakly composable}\/.
A vast majority of the most popular generalized entropies in the
literature are at least weakly composable.

The notion of group entropy is a direct consequence of the previous
discussion. In the usual statistical picture, we shall describe a system possessing a certain number $W$ of microstates by a probability density, which is a $W$-component
vector of nonnegative reals satisfying
\[
p=(p_1,\dots, p_W)\,,\qquad p_i\geq0\,,\qquad \sum_{i=1}^W p_i=1\,.
\]
Hence $p_i$ is the probability that the system~$A$ is in the $i\th$ microstate.
%\begin{definition}
A \textit{group entropy} is a nonnegative function $S(p_1,\ldots,p_W)$ that satisfies the axioms (SK1)--(SK3) and is strictly composable.
%\end{definition}
In other words, in addition to the usual requirements of continuity,
expansibility and concavity, a group entropy (due to the properties
\eqref{composable}-\eqref{group-structure}) also possesses a
group-theoretical structure represented by the product~$\Phi(x,y)$ and
which holds in {all} the probability
distribution space associated with a given complex system.

Prime examples of group entropies are the Boltzmann and R\'enyi entropies, which satisfy~\eqref{composable} with the additive law $\Phi(x,y)=x+y$. Non-additive laws have been studied extensively
too, since they arise naturally in the context of complex systems, where the entropy of the total system is expected to different from the
sum of the entropies of the independent parts. Tsallis's entropy, for
instance, obeys the law $\Phi(x,y)= x+y+ (1-q) xy$, corresponding to
the multiplicative formal group law. Although we will not rely on
these techniques for this paper, it worth mentioning that formal group
theory \cite{Haze}, intensively investigated in algebraic topology
since the second half of XX century, provides a natural classification
of group laws in terms of formal power series, as it offers very general algebraic results and a natural language to formulate the theory of generalized entropies.

The main problem that we address in this paper is to ascertain which
is the most general form that a strictly composable generalized
entropy can take. For the sake of generality, in the subsequent
analysis we will not even assume that $\Phi(x,y)$ be a group law: we
shall only assume that it is an arbitrary two-variable function, with
suitable regularity properties, and derive from first principles that, in particular, it has
to be a group law.

Before stating our results, let us recall the expression of the
Boltzmann--Gibbs entropy of a system in a state described by a
probability distribution $p=(p_1,\dots, p_W)$  reads
$$
S\sub{BG}(p):=\sum_{i=1}^{W} p_i\,\ln\frac1{p_i}\,.
$$
(Throughout this work, we will set Boltzmann's constant to one: $k_{\mathrm{B}}=1$.) Inspired by this expression, perhaps the most common way of constructing entropy functions is in {\em trace form}, which in the above notation means that there is a
one-variable nonnegative function
$f(t)$ such that
\begin{subequations}\label{Sf}
\begin{equation}
S(p)=\sum_{i=1}^W f(p_i)\, ,
\end{equation}
with the constraints
\begin{equation}\label{f0f1}
f(0)=f(1)=0\, .
\end{equation}
\end{subequations}
This condition ensures that the entropy is zero if the probability
distribution of a system is $p_i=\de_{i,1}$ (that is, the system is in
a certainty state). If a quantum system is described by a density matrix~$\rho$, then the corresponding quantum entropy can be directly computed as $\tr f(\rho)$.

Since Shannon's foundational paper in information
theory~\cite{Shannon}, an extensive body of literature on trace-form
entropies has appeared. The prototype example of trace-form generalized entropy is the one-parameter generalization of Boltzmann's entropy introduced by Tsallis, given for $q\geq0$ and $q\neq1$ by the formula
\begin{equation}\label{Sq}
S_q(p):=\sum_{i=1}^Wf_q(p_i)\,,\qquad f_q(t):=\frac{t-t^q}{q-1}\,.
\end{equation}
As $q$ tends to $1$ one recovers the Boltzmann--Gibbs entropy, so it
is customary to set $S_1(p):=S\sub{BG}(p)$ and $f_1(t):=t\ln\frac1t$. A two-parameter presentation of the Tsallis entropy was recently introduced in~\cite{PPROCA16}.

Our first result shows that, in a way, the class of trace-form
entropies has a serious drawback: we prove that, under mild regularity
assumptions, \textit{the only composable trace-form entropy is the
  Tsallis entropy}\/, with Boltzmann-Gibbs as a particular
case. Consequently, in order to construct new entropies one must
either assume that they are not in trace form or to deal with the fact
that if one has two independent systems $A$ and $B$, the entropy of
the total system will not be determined, in general, in terms of the
entropies of the independent subsystems. Of course, this also implies that no group-theoretical structure is available. % with unpleasant thermodynamic and information-theoretical consequences \cite{ST16}
The result can be stated as follows.

\begin{theorem}\label{T.main}
  Let $S$ be an entropy of the form~\eqref{Sf} with a
  function $f$ of class $C^2((0,1))\cap
  C^1([0,1])$. Suppose that
  these entropies satisfy the
  condition~\eqref{composable} with a composition law $\Phi$ of class
  $C^1$. Then there are positive
  real constants $c,q$ such that
\[
f(t)=c\, f_q(t)\,,
\]
so $S$ is the Tsallis entropy~\eqref{Sq} for some real~$q$, up to a
multiplicative constant. The composition law is
$\Phi(x,y)=x+y+\al xy$ for some explicit constant~$\al$ depending on~$c$ and~$q$.
\end{theorem}

At the same time, it is worth stressing that not all the entropies commonly employed in information theory are of trace form. For example, the celebrated R\'enyi entropy \cite{Renyi1960}
\[
S\sub{R}(p):=\frac1{1-\al}\ln\bigg(\sum_{i=1}^W p_i^\al\bigg), \qquad \alpha>0
\]
is indeed not in this class. Hence, inspired by the form of R\'enyi's entropy, it is natural to consider generalized entropies of the form
\begin{subequations}\label{S2}
\begin{equation}
\tS(p)=g\Bigg(\sum_{i=1}^W h(p_i)\Bigg)\,,
\end{equation}
where
\begin{equation}
h(0)=g(h(1))=0\,.
\end{equation}
\end{subequations}
These conditions once again ensure that the entropy of a system with probability distribution $p_i=\de_{i1}$ is zero.
The function $g$ is typically assumed to be monotone. Note that the
class \eqref{S2}, under the additional constraint that $h$ be a
concave function, was first considered in \cite{SMMP93}. A quantum
version of this family of entropies was studied in \cite{BZHPL16}.  % It
% is should be observed that, just as in the case of trace-form
% entropies, if the system is described by a density matrix~$\rho$, the
% entropy can be simply computed as $g(\tr h(\rho))$.
Nontrivial examples of entropy functions of this form are provided by the large class of $Z$-entropies, recently introduced in
\cite{PPROCA16}. They  are multiparametric non-trace-form {group
  entropies} of the form \eqref{S2} which, under mild assumptions,
reduce to the standard R\'enyi entropy in a suitable limit.

Our second result is a characterization of the generalized entropies
of the form~\eqref{S2} that are composable. Again for the sake of
generality, we do not assume that the function~$g$ is convex. Our result essentially asserts that, for any given monotone function~$g$, a function of the form \eqref{S2} satisfies the composability condition if
and only if the ``trace part'' is of the form $h(t)=at+bt^q$ for some
real constants. More precisely, we have the following statement:

\begin{theorem}\label{T.main2}
  Let $S$ be an entropy of the form~\eqref{S2}, where
  the function $h$ is of class $C^2((0,1))\cap
  C^1([0,1])$ and $g$ is a $C^1$ function
  with $g'\neq0$. Then
  the entropy~$S$ satisfies the composability
  condition~\eqref{composable} with $\Phi$ of class $C^1$ if and only if
\[
h(t)=at+bt^q
\]
for some real constants $a,b$ and $q>1$. The composition law can be written
down explicitly in terms of these constants and the function~$g$.
\end{theorem}

It should be remarked that the regularity hypothesis is probably not sharp, but it is crucially used in a rather tricky derivation of
differential equations from some (rather unmanageable) functional equations that lies at the core of the proof of the theorems. Notice
that the functions $f(t)=(t-t^q)/(q-1)$ and $h(t)=t^\al$ that respectively appear in the Tsallis and R\'enyi entropies are of
class~$C^1$ at~0 precisely for $q>1$ and $\al \geq1$, and that we are
not obtaining the function $h(t)=at+bt\ln\frac1t$ because we require
$h$ to be continuously differentiable at~$0$.

Let us conclude the Introduction with some comments about the proof of
these results and the organization of the paper. The proof of both
theorems, respectively presented in Sections~\ref{S.proof}
and~\ref{S.proof2}, are based on similar arguments, so let us
illustrate them in the case of Theorem~\ref{T.main}. The proof of this
result involves three ideas. Firstly, an easy argument shows that the
function~$\Phi$ appearing the composition law must be an associative,
commutative product. With some more work, which involves taking
variations in the equation with respect to the probability densities,
we show that in fact the only admissible composition function is
$\Phi(x,y)=x+y+\al xy$ for some real constant~$\al$. A key tool to
prove this will be a lemma on the functional independence of certain
functions that we present in Section~\ref{S.lemma}. It is worth
mentioning that the result on the structure of the composition
function holds
with less stringent regularity hypotheses (specifically, for any $f$
absolutely continuous on~$[0,1]$). In passing from the expression
for~$\Phi$ to the general form of~$f$ we must indeed use the
$C^1$~regularity of~$f$ up to the endpoints of the interval to get rid
of some of the several parameters that our argument relies on. Some
further observations in this direction are presented in
Section~\ref{S.remarks}.

\section{A lemma on the functional dependence of traces}
\label{S.lemma}

In this section we will prove a key lemma on the structure of the
functions that satisfy a certain kind of functional relations which is
strongly related with the composability
condition~\eqref{composable}.

To motivate this result, it is convenient to start
by recalling a basic definition about independent systems that will be
used in the rest of the paper.
Given two probability distributions
\begin{subequations}\label{products}
\begin{equation}
p^A=(p_i^A)_{i=1}^W \quad \text{and} \quad p^B=(p^B_j)_{j=1}^{W'}\,,
\end{equation}
their number of states respectively being $W$ and $W'$, by the total
system $A\cup B$ we mean the $WW'$~state systems described by the probability
distribution
\begin{equation}
p^{A\cup B}=(p_{ij}^{A\cup B})_{1\leq i\leq W,1\leq j\leq W'}\,,
\end{equation}
with
\begin{equation}%\label{pij}
p_{ij}^{A\cup B}:= p^A_i p^B_j\,.
\end{equation}
\end{subequations}
It is this formula, together with the expressions~\eqref{Sf}
and~\eqref{S2} for the entropies that we will study in this paper, which leads us next to consider expressions of the
form
\[
\sum_{i=1}^W\sum_{j=1}^{W'} F(p_i^Ap_j^B)\,.
\]

The following lemma, which is the main result of this section, asserts that
if the functions
\[
\sum_{i=1}^W\sum_{j=1}^{W'} F(p_i^Ap_j^B)\,,\quad \sum_{i=1}^W F(p_i
^A)\quad\text{and}\quad \sum_{j=1}^{W'}F(p_j^B)
\]
are functionally dependent (as functions of the probability densities
$p^A$ and $p^B$), then the expression of the first of these functions
in terms of the other two is given by a very simple algebraic relation
that only depends on four parameters.

\begin{lemma}\label{L.Phi}
Suppose that there is a one-variable function $F\in C^1((0,1))$ and a
$C^1$ function $\Psi$ of two variables such that, for any probability densities $p^A$ and $p^B$ as
above, one has
\begin{equation}\label{eq}
\sum_{i=1}^W\sum_{j=1}^{W'} F(p_i^Ap_j^B)=\Psi\bigg(\sum_{i=1}^W F(p_i ^A),\sum_{j=1}^{W'}F(p_j^B)\bigg)\,.
\end{equation}
Then necessarily
\[
\Psi(x,y)=a_0+a_1x+a_2y+a_3 xy
\]
for some real constant~$a_j$.
\end{lemma}

\begin{proof}
We shall next take variations in the identity~\eqref{eq} with respect to
the probability density $p^A$. Notice that the space of probability
densities with $W$ states is
\begin{align}\label{set}
\mathcal P_W&:=\bigg\{p=(p_1,\dots, p_W)\in\RR^W:p_i\geq 0,\; \sum_{i=1}^W
  p_i=1\bigg\}\,,
\end{align}
which one can understand as a bounded subset of $\RR^{W-1}$ with
nonempty interior.
% which one can parametrize by the first $W-1$~components of the
% probability vector as follows:
% \begin{multline*}
% \mathcal P_W=\bigg\{p=(p_1,\dots,
% p_W)\in\RR^W:  0\leq p_i\leq
% 1-\sum_{l=1}^{i-1} p_l\; \text{ for }1\leq i\leq W-1,\\
% p_W=1-\sum_{l=1}^{W-1}p_l\bigg\}\,.
% \end{multline*}
To consider variations of the probability
density~$p^A$, let us choose without loss of generality a probability
density $p^A$ that does not lie on the boundary of the
set~\eqref{set}, i.e., such that $p^A_i>0$ for all~$i$. Then the curve
\begin{equation}\label{p(s)}
p(s)=\bigg(p^A_1+ sP_1,p^A_2+sP_2,\dots, p^A_{W-1}+sP_{W-1},p^A_{W-1}-s\sum_{l=1}^{W-1}P_l\bigg)
\end{equation}
is contained in the set~\eqref{set} for an arbitrary vector
\[
P:=(P_1,\dots,P_{W-1})
\]
in $\RR^{W-1}$ with $|P|\leq 1$, and for all~$s$ in a small enough
interval $s\in(-\ep,\ep)$.

Inserting this curve in the identity~\eqref{eq} we obtain
\begin{equation}\label{eq2}
\sum_{i=1}^W\sum_{j=1}^{W'} F(p_i(s)\,p_j^B)-\Psi\bigg(\sum_{i=1}^W
F(p_i(s)),\sum_{j=1}^{W'}F(p_j^B)\bigg)=0
\end{equation}
for all $s\in (-\ep,\ep)$ and all $P$ in the unit ball of
$\RR^{W-1}$. Differentiating this relation at $s=0$ one then obtains
\begin{multline}
\sum_{l=1}^{W-1}\bigg(\sum_{j=1}^{W'}p_j^B\,\big[F'(p_l^Ap_j^B)-F'(p_W^Ap_j^B)\big]\\
-
D_1\Psi\bigg(\sum_{i=1}^W F(p_i^A),\sum_{j=1}^{W'} F(p_j^B)\bigg)\,
\big[ F'(p_l^A)-F'(p_W^A)\big]\bigg)\, P_l=0\,,\label{sumW0}
\end{multline}
where $D_1\Psi$ denotes the derivative of the function $\Psi$ with
respect to first argument. Given that this relation must hold for all
$P\in\RR^{W-1}$ with $|P|\leq 1$, we infer that for all
$1\leq l\leq W-1$ one has
\begin{multline}\label{sumW}
\sum_{j=1}^{W'}p_j^B\,\big[F'(p_l^Ap_j^B)-F'(p_W^Ap_j^B)\big] \\
=D_1\Psi\bigg(\sum_{i=1}^W F(p_i^A),\sum_{j=1}^{W'} F(p_j^B)\bigg)\,
\big[ F'(p_l^A)-F'(p_W^A)\big]
\end{multline}

As the left hand side of the identity~\eqref{sumW} does
not depend on $p_i^A$ for $i\neq l, W$, choosing without loss of
generality $W\geq 4$ (so that the dimension of the space of
probability densities is at least~3) it follows that so must be the
right hand side. Hence we infer that
\begin{equation}\label{eq3}
D_1\Psi\bigg(\sum_{i=1}^W F(p_i^A),\sum_{j=1}^{W'} F(p_j^B)\bigg)=a\bigg( \sum_{j=1}^{W'} F(p_j^B)\bigg)
\end{equation}
for some function~$a$ or, to put it differently,
\[
\frac\pd{\pd x}\Psi(x,y)=a(y)\,.
\]
This can be immediately integrated
to yield
\begin{equation}\label{F1}
\Psi(x,y)=a(y) x + b(y)\,,
\end{equation}
with $b$ another arbitrary function.

One can reserve the role of $p^A$ and $p^B$ and consider variations of
the identity~\eqref{eq} with respect to the probability density
$p^B$. Arguing as above we then infer that
\[
\frac\pd{\pd y}\Psi(x,y)=\widetilde a(x)
\]
for some function $\widetilde a$,
or equivalently
\begin{equation}\label{F2}
\Psi(x,y)=\widetilde a(x) y + \widetilde b(x)\,,
\end{equation}
with $\widetilde b$ another arbitrary function. From~\eqref{F1}
and~\eqref{F2} we obtain that $\Psi(x,y)$ is a polynomial of order~1 both
in~$x$ and~$y$ (separately), so
\[
\Psi(x,y)=a_0+a_1x+a_2y+a_3xy
\]
for some real constants $a_j$.
\end{proof}

\section{Proof of Theorem~\ref{T.main}}
\label{S.proof}

In this section we present the proof of the
theorem, which consists of two steps.

\subsubsection*{Step 1: The composition function is $\Phi(x,y)=x+y+\al
  xy$}

Using the notation~\eqref{products} introduced in the previous section, the composability
condition~\eqref{composable} reads as
\begin{equation}\label{eq*}
\sum_{i=1}^W\sum_{j=1}^{W'} f(p_i^Ap_j^B)=\Phi\bigg(\sum_{i=1}^W f(p_i ^A),\sum_{j=1}^{W'}f(p_j^B)\bigg)\,.
\end{equation}
Since this relation must hold for all probability distributions $p^A$
and $p^B$, Lemma~\ref{L.Phi} ensures that the composition law must be
of the form
\begin{equation}\label{Phi0}
\Phi(x,y)=a_0+a_1x+a_2y+a_3xy
\end{equation}
for some real constants $a_j$.

Let us now evaluate~\eqref{eq*} when the second probability
distribution is $p^B_j=\de_{j1}$. As $f(0)=f(1)=0$, we then get that
for any $p^A$ we have
\[
\sum_{i=1}^W f(p_i^A)=\Phi\bigg(\sum_{i=1}^W f(p_i ^A),0\bigg)\,,
\]
which means that
\begin{equation}\label{Phi1}
\Phi(x,0)=x\,.
\end{equation}
Taking now $p^A_i=\de_{i1}$ and an arbitrary $p^B$ we similarly obtain
\[
\Phi(0,y)=y\,,
\]
which together with~\eqref{Phi1} ensures that the only function
$\Phi(x,y)$ of the form~\eqref{Phi0} that one can have here is
\begin{equation}\label{formPhi}
\Phi(x,y)=x+y+\al xy
\end{equation}
with $\al$ a real constant.

\subsubsection*{Step 2: The general form of~$f(t)$ is that of Tsallis entropy}

To find the expression for~$f$, let us consider variations in the
identity~\eqref{eq*}, just as in the proof of
Lemma~\ref{L.Phi}. Indeed, substituting the probability density~$p^A$
by the curve $p(s)$ defined in~\eqref{p(s)} and differentiating at
zero, we obtain that
\begin{multline*}
\sum_{l=1}^{W-1}\bigg(\sum_{j=1}^{W'}p_j^B\,\big[f'(p_l^Ap_j^B)-f'(p_W^Ap_j^B)\big]\\
-
\bigg(1+\al\sum_{j=1}^{W'} f(p_j^B)\bigg)\,
\big[ f'(p_l^A)-f'(p_W^A)\big]\bigg)\, P_l=0\,.%\label{sumW0}
\end{multline*}
Here we have used that the function~$\Phi$ is given by~\eqref{formPhi}.
As the constants $P_l$ are arbitrary, this shows that for all $1\leq
l\leq W-1$ one has
\begin{equation}\label{eq4}
\sum_{j=1}^{W'}p_j^B\,\big[f'(p_l^Ap_j^B)-f'(p_W^Ap_j^B)\big]=
\bigg(1+\al\sum_{j=1}^{W'} f(p_j^B)\bigg)\,
\big[ f'(p_l^A)-f'(p_W^A)\big]\,.
\end{equation}

Let us now consider variations with respect to the probability
density~$p^B$. Just as in the proof Lemma~\eqref{L.Phi}, let us assume that $p^B$ is
not on the boundary of the set of $W'$-state probability densities
$\mathcal P_{W'}$ (i.e., $p_ j^B>0$ for all $1\leq j\leq W'$). (We
recall that the set $\mathcal P_{W'}$ was defined in~\eqref{set}). Then one can take a small
enough $\ep$ such that the curve
\begin{equation}\label{barP}
\bp(s):=\bigg(p_1^B+s\bP_1, p_2^B+s\bP_2,\dots, p_{W'-1}^B+s\bP_{W'-1}, p^B_{W'}-s\sum_{m=1}^{W'-1}\bP_m\bigg)
\end{equation}
is contained in $\mathcal P_{W'}$ for all $s\in (-\ep,\ep)$ and each vector
\[
\bP=(\bP_1,\dots, \bP_{W'-1})
\]
in $\RR^{W'-1}$ with $|\bP|\leq 1$.

Evaluating the identity~\eqref{eq4} on $p^B=\bp(s)$ and differentiating at
$s=0$ we then get that for all $1\leq l\leq W-1$ one has
\begin{multline}\label{eqbarP}
  \sum_{m=1}^{W'-1}\bigg[f'(p_l^Ap_m^B)-f'(p_l^Ap_{W'}^B)-f'(p_W^Ap_m^B)+f'(p_W^Ap_{W'}^B)+p_l^Ap_m^Bf''(p_l^Ap_m^B)\\
  - p_l^Ap_{W'}^Bf''(p_l^Ap_{W'}^B)-p_W^Ap_m^Bf''(p_W^Ap_m^B)+
  p_W^Ap_{W'}^Bf''(p_W^Ap_{W'}^B)\\
-\al[f'(p_m^B)-f'(p_{W'}^B)][f'(p_l^A)-f'(p_W^A)]\bigg]\,
  \bP_m=0\,.
\end{multline}
Since this holds for all $\bP$ in the unit ball of $\RR^{W'-1}$, this
ensures that for all $1\leq l\leq W-1$ and $1\leq m\leq W'-1$ one has
\begin{multline}\label{eq5}
f'(p_l^Ap_m^B)-f'(p_l^Ap_{W'}^B)-f'(p_W^Ap_m^B)+f'(p_W^Ap_{W'}^B)+p_l^Ap_m^Bf''(p_l^Ap_m^B)\\
  - p_l^Ap_{W'}^Bf''(p_l^Ap_{W'}^B)-p_W^Ap_m^Bf''(p_W^Ap_m^B)+
  p_W^Ap_{W'}^Bf''(p_W^Ap_{W'}^B)
\\
=\al[f'(p_m^B)-f'(p_{W'}^B)][f'(p_l^A)-f'(p_W^A)]\,.
\end{multline}

Without loss of generality let us take
$W\geq3$ to ensure that the variables $p_i^A,p_W^A$
are independent. To transform the functional equation~\eqref{eq5} into
a differential equation, it is convenient
to evaluate this identity on the probability distribution
$p^B_m=\de_{m1}$. (Notice that, although we have used that $p^B_j>0$
for all~$j$ to derive the equation, by continuity it must also hold
for this choice of $p^B$.) Let us fix a certain~$l$ and write
\[
t:=p^A_l\,,\qquad \tau:=p^A_W\,,
\]
Equation~\eqref{eq5} then reads as
\begin{equation}\label{eq6}
tf''(t)+(1-q) f'(t)= \tau f''(\tau)+(1-q)  f'(\tau)\,,
\end{equation}
where we have used the fact that $f'$ is continuous up to the
endpoints of the interval $[0,1]$ to set
\[
q:=\al(f'(1)-f'(0))\,.
\]

Since $t$ and $\tau$ are independent, Equation~\eqref{eq6} implies that
\[
tf''(t)+(1-q)  f'(t)=-c\,,
\]
where $c$ is a constant. For $q>0$, the solution of this equation is given in
terms of two arbitrary constants as
\begin{equation}\label{formf}
f(t)=\frac{c t}{q-1}+c_1 t^{q}+c_2 \,,
\end{equation}
so the conditions $f(0)=f(1)=0$ then imply that
\[
f(t)=c\frac{t-t^q}{q-1}
\]
where the normalization constant $c$ remains arbitrary. The case
$q\leq1$ leads to functions that are not of class $C^1$ at
$0$. Theorem~\eqref{T.main} then follows.

\section{Proof of Theorem~\ref{T.main2}}
\label{S.proof2}

In this section we present the proof of Theorem~\ref{T.main2}, which
relies on the same kind of ideas as that of Theorem~\ref{T.main}. Just as
before, it is convenient to divide the proof in two steps:

\subsubsection*{Step 1: Derivation of the composition law}

The
starting point is the composability condition~\eqref{composable},
which using the notation~\eqref{products} and the form of the metric
one can write as
\[
g\bigg(\sum_{i=1}^W\sum_{j=1}^{W'}
h(p_i^Ap_j^B)\bigg)=\Phi\bigg(g\bigg(\sum_{i=1}^W h(p_i ^A)\bigg) ,g\bigg(\sum_{j=1}^{W'}h(p_j^B)\bigg)\bigg)\,.
\]
As the function $g$ has a $C^1$ inverse $g^{-1}$ because $g'\neq0$,
one can define the $C^1$ function
\[
\widetilde \Phi(x,y):=g^{-1}(\Phi(g(x),g(y)))\,,
\]
in terms of which the above relation reads as
\begin{equation}\label{eqF}
\sum_{i=1}^W\sum_{j=1}^{W'}
h(p_i^Ap_j^B)=\widetilde \Phi\bigg(\sum_{i=1}^W h(p_i ^A) ,\sum_{j=1}^{W'}h(p_j^B)\bigg)\,.
\end{equation}
Since this identity holds true for any probability densities $p^A$ and
$p^B$, Lemma~\ref{L.Phi} then ensures that there are constants $a_j$
such that
\begin{equation}\label{formF}
\widetilde \Phi(x,y)=a_0+a_1x+a_2y+a_3 xy\,.
\end{equation}

To compute the values of the constants~$a_j$, let us now evaluate the
identity~\eqref{eqF} when $p^B_j=\de_{j1}$. Since $h(0)=0$, letting
\[
\be:=h(1)
\]
we then obtain that
\[
\sum_{i=1}^W h(p_i ^A)=\widetilde \Phi\bigg(\sum_{i=1}^W h(p_i ^A),\be\bigg)\,,
\]
for any probability density $p^A$, that is,
\begin{equation}\label{F1}
\widetilde \Phi(x,\be)=x\,.
\end{equation}
If we now take the probability density $p^A_i=\de_{i1}$ and an
arbitrary $p^B$, we analogously arrive at
\begin{equation}\label{F2}
\widetilde \Phi(\be,y)=y\,.
\end{equation}
A straightforward computation then shows that the only functions of the
form~\eqref{formF} that satisfy~\eqref{F1} and~\eqref{F2} are
\[
\widetilde \Phi(x,y)=x+y-\be+\al(x-\be)(y-\be)\,,
\]
where $\al$ is an arbitrary real constant. This shows that the
composition law is
\begin{equation}\label{Phi-final}
\Phi(x,y)=g\big(g^{-1}(x)+ g^{-1}(y)-\be+\al(g^{-1}(x)-\be)(g^{-1}(y)-\be)\big)\,.
\end{equation}

\subsubsection*{Step 2: The form of the function~$h$}

Here we shall proceed by taking variations just as in Step~2 of the
proof of Theorem~\ref{T.main}. Let us sketch the details.

First we take variations with respect to the probability density
$p^A$, so we replace $p^A$ by the curve $p(s)$ (cf.\
Equation~\eqref{p(s)}) in the identity~\eqref{eqF}. Taking derivatives
at $s=0$ and using the explicit form of $\widetilde \Phi$
(Equation~\eqref{Phi-final}) we obtain that
\begin{multline*}
\sum_{l=1}^{W-1}\bigg(\sum_{j=1}^{W'}p_j^B\,\big[h'(p_l^Ap_j^B)-h'(p_W^Ap_j^B)\big]\\
-
\bigg(1-\al\be+\al\sum_{j=1}^{W'} h(p_j^B)\bigg)\,
\big[ h'(p_l^A)-h'(p_W^A)\big]\bigg)\, P_l=0\,.%\label{sumW0}
\end{multline*}
As the constants $P_l$ are arbitrary, this shows that for all $1\leq
l\leq W-1$ one has
\begin{equation}\label{eq4p}
\sum_{j=1}^{W'}p_j^B\,\big[h'(p_l^Ap_j^B)-h'(p_W^Ap_j^B)\big]=
\bigg(1-\al\be+\al\sum_{j=1}^{W'} h(p_j^B)\bigg)\,
\big[ h'(p_l^A)-h'(p_W^A)\big]\,.
\end{equation}
This is the analog of Equation~\eqref{eq4}.

Now we take variations with respect to $p^B$ in~\eqref{eq4p}. That is,
we replace $p^B$ by the curve $\bar p(s)$ introduced in~\eqref{barP} and
take the derivative at $s=0$ to obtain that for all $1\leq l\leq W-1$ one has
\begin{multline}\label{eqbarP}
  \sum_{m=1}^{W'-1}\bigg[h'(p_l^Ap_m^B)-h'(p_l^Ap_{W'}^B)-h'(p_W^Ap_m^B)+h'(p_W^Ap_{W'}^B)+p_l^Ap_m^Bh''(p_l^Ap_m^B)\\
  - p_l^Ap_{W'}^Bh''(p_l^Ap_{W'}^B)-p_W^Ap_m^Bh''(p_W^Ap_m^B)+
  p_W^Ap_{W'}^Bh''(p_W^Ap_{W'}^B)\\
-\al[h'(p_m^B)-h'(p_{W'}^B)][h'(p_l^A)-h'(p_W^A)]\bigg]\,
  \bP_m=0\,.
\end{multline}
This is exactly Equation~\eqref{eqbarP}, but with $h$
playing the role of~$f$. Hence we infer from Equation~\eqref{formf}
that $h$ must be of the form
\[
h(t)=at+b t^{q}+c\,,
\]
where $a,b,c$ are real constants. As $h(0)=0$, we must have $c=0$,
which completes the proof of the theorem.

\section{The cases of absolutely continuous or analytic functions}
\label{S.remarks}

In this section we will present a couple of remarks about the regularity assumptions in our results. As the proofs of both theorems
involve essentially the same ideas, for concreteness we will make this remarks only in the context of the first theorem (that is, trace-form
entropies); the extension to entropies of the form~\eqref{S2} is straightforward.

We have determined the form of the composition law, which is $\Phi(x,y)=x+y+\al xy$ by means of
Theorem~\ref{T.main} and is given by~\eqref{Phi-final} in
Theorem~\ref{T.main2}) through Lemma~\ref{L.Phi}. The first
observation is that this lemma holds under considerably weaker
regularity assumptions, and that in fact we do not even need to assume
that the composability condition holds for all probability densities:
if it holds in any open subset of the space of probability densities
(for example), for densities that are close enough to the uniform
distributions $p^A_i=1/W$, $p^B_j=1/W'$), the argument goes
through. More precisely, we have the following:

\begin{proposition}
Let $S$ be an entropy of the form~\eqref{Sf} satisfying the
composability condition~\eqref{composable} in a small neighborhood of
any two probability densities $\hat p^A$ and $\hat p^B$. If $f$ is absolutely
continuous on $[0,1]$ and $\Phi$ is of class~$C^1$, then in the
neighborhood under consideration the
composition law must be $\Phi(x,y)=x+y+\al xy$ for some real constant~$\al$.
\end{proposition}

\begin{proof}
Since the argument we used in Lemma~\eqref{L.Phi} to derive the
differential equations from functional relations is purely local
(because it relies on taking curves $p(s)$ with $s\in (-\ep,\ep)$), it
is clear that the fact that the composability condition only holds in
an open set does not constitute a problem.

The only aspect that one must control to derive the proposition is to
ensure that the proof of Lemma~\ref{L.Phi} also remains valid under
the weaker regularity assumption that $f \in AC([0,1])$ (of course,
our function~$f$ will play the role of the function~$F$ in the lemma). The key point
is to make sense of Equation~\eqref{sumW0} (or, equivalently, \eqref{sumW}). This is formally the derivative at $s=0$
of the map given by the left hand side of~\eqref{eq2}, so our goal is
to make sense of it as a differentiable function of $s$. This is not
hard, but it does not follow from a general distribution-theoretical
argument either.

To this end, recall that the derivative of an absolutely continuous
function is $f'\in L^1((0,1))$. The key feature is that in
Equation~\eqref{sumW0} one does not have to deal with functions of the
form, say, $F'(p_i^A)$, but with $p_j ^B \,F'(p_i ^A p_j^B)$. This
is important because, although the fact that a function $h(t)$ is in $L^1((0,1))$ does
not imply that $h(xy)$ is in $L^1((0,1)\times (0,1))$ (this can be readily
seen by taking $h(t):=\frac1t(\ln\frac 2t)^{-2}$, for instance),
setting $z:=xy$ one sees that
\[
\int_0^1 \int_0^1 x\, |h(xy)|\, dx\, dy= \int_0^1 \int_0^y  |h(z)|\,
dz\, dy\leq \int_0^1|h(z)|\, dz\,.
\]
This shows that $x\, h(xy)\in L^1((0,1)\times (0,1))$ whenever $h(t)\in
L^1((0,1))$.

Getting back to our problem and recalling that
\begin{equation}\label{pWpW'}
p_W^A=1-\sum_{i=1}^{W-1} p_i^A\,,\qquad p_{W'}^B=1-\sum_{i=1}^{W'-1} p_i^B
\end{equation}
can be written in terms of the remaining $W-1$ (respectively $W'-1$)
components of the vector, we then infer that with $p_W^A$ and
$p_{W'}^B$ given by~\eqref{pWpW'},
\begin{multline*}
G\big((p_i^A)_{i=1}^{W-1},(p_j^B)_{j=1}^{W'-1}\big):=\sum_{l=1}^{W-1}\sum_{j=1}^{W'}p_j^B\,\bigg[f'(p_l^Ap_j^B)-f'(p_W^Ap_j^B)\big]\\
-
D_1\Phi\bigg(\sum_{i=1}^W f(p_i^A),\sum_{j=1}^{W'} f(p_j^B)\bigg)\,
\big[ f'(p_l^A)-f'(p_W^A)\big]\bigg)\, P_l
\end{multline*}
defines a map
\[
G\in L^1((0,1)^{W+W'-2})\,,
\]
which depends linearly on the parameters $P_l$. Of course, here we are
using that the function
\[
D_1\Phi\bigg(\sum_{i=1}^W f(p_i^A),\sum_{j=1}^{W'} f(p_j^B)\bigg)
\]
is continuous, so its product with an $L^1$ function is
in $L^1$.
Hence, with $s\in(-\ep,\ep)$, the right hand side of
Equation~\eqref{eq2} defines a map $g(s)$ such that
\[
g\in C ((-\ep,\ep), C((0,1)^{W+W'-2}))\cap C^1((-\ep,\ep), L^1((0,1)^{W+W'-2}))
\]
with derivative $g'(0)=G$. From this it stems that the proof of Step~1
does remain valid for a general $f\in AC([0,1])$.
\end{proof}

The second observation is that the proof of Step~2 in Theorem~\ref{T.main} (that is,
passing from the identity $\Phi(x,y)=x+y+\al xy$ to the general form of $f(t)$)
becomes much easier if we assume that $f$ is are
analytic in $[0,1]$ (of course, this will not be case in general: in fact,
the function appearing in the Tsallis entropy is only analytic at $0$
when $q$ is an integer greater than or equal to~2).

In this simple case, the fact that necessarily $f(t)=c(t-t^q)$ for some constants $c$ and $q$ can be derived from a general result on
functional equations due to Aczel~\cite{Aczel}. Actually, a simple proof of this can be given directly from the composition equation
\[
\sum_{i=1}^W\sum_{j=1}^{W'} f(p_i^Ap_j^B)=\sum_{i=1}^W f(p_i ^A) +
\sum_{j=1}^{W'}f(p_j^B) + \al\sum_{i=1}^W \sum_{j=1}^{W'} f(p_i ^A)\,f(p_j^B)=0\,.
\]
We will present it for the benefit of the reader. We start by considering the uniform probability distributions
\[
p^A_i=\frac1 W\,,\qquad p^B_j=\frac1{W'}
\]
with arbitrarily large numbers of states~$W,W'$. We then have that, setting
$h(t):=f(t)/t$,
\[
h(W^{-1}W'^{-1})=h(W^{-1})+ h(W'^{-1})+\al h(W^{-1})h(W'^{-1})\,.
\]
As the sequences $(W^{-1})_{W=1}^\infty$ and $(W'^{-1})_{W'=1}^\infty$
tend to~$0$, the analyticity of $h$ on $[0,1]$ implies that for all
$s,t\in [0,1]$ one has
\[
h(st)=h(s)+h(t)+\al h(s)h(t)\,.
\]
Hence one can differentiate with respect to the variable~$s$ and
evaluate at $s=1$ to find
\[
t\,h'(t)=\ga\,\big(1+\al h(t)\big)\,,
\]
with $\ga:=h'(1)$, which can be readily integrated to obtain that
\[
h(t)=-\frac1\al +Ct^{-\al\ga}\,.
\]
Since $h(1)=0$, this readily gives
\[
h(t)=c(t^\nu-1)\,,
\]
as we wanted to prove.

\section*{Acknowledgements}

A.E.\ is supported by the ERC Starting Grant~633152 and by the ICMAT--Severo Ochoa grant SEV-2015-0554 (MINECO). P. T. has been partly supported by the research project FIS2015-63966, MINECO, Spain, Spain, and by the ICMAT Severo Ochoa grant SEV-2015-0554.

 \end{document}